\documentclass[11pt]{article}
\usepackage[letterpaper,margin=1in]{geometry}
\usepackage[T1]{fontenc}
\usepackage{lmodern}
\usepackage{amsmath,amssymb,amsthm,mathtools}
\usepackage{booktabs,tabularx,array}
\usepackage{microtype}
\usepackage{needspace}
\usepackage{tikz}
\usetikzlibrary{arrows.meta,positioning}
\usepackage[colorlinks=true,linkcolor=blue!45!black,citecolor=blue!45!black,urlcolor=blue!45!black]{hyperref}
\usepackage[nameinlink,noabbrev,capitalize]{cleveref}
\newtheorem{theorem}{Theorem}
\newtheorem{proposition}[theorem]{Proposition}
\newtheorem{lemma}[theorem]{Lemma}

\theoremstyle{definition}

\newcommand{\F}{\mathbb F}
\newcommand{\Z}{\mathbb Z}

\newcommand{\ket}[1]{\lvert #1\rangle}

\newcommand{\wt}{\operatorname{wt}}
\newcommand{\supp}{\operatorname{supp}}
\newcommand{\Tr}{\operatorname{Tr}}

\newcommand{\spanof}{\operatorname{span}}
\newcommand{\diag}{\operatorname{diag}}

\DeclareMathOperator{\End}{End}
\title{Constant Rate Codes with Fully Addressable Transversal T: Good Codes, Sparse Checks}
\author{Alexander Meiburg\footnote{transt@ohaithe.re}\\{\em University of Waterloo, Department of Applied Math}\\{\em Perimeter Institute for Theoretical Physics}}
\date{September 29 2026}
\begin{document}
\maketitle
\begin{abstract}
Recent work by Wills~\cite{Wills} constructed codes with constant rate and transversal $T$.
Extending this line of work, we construct asymptotically good CSS codes with
\textbf{fully addressable} transversal $T$, permitting selective control over
which logical qubits receive a $T$ gate.
We also give Pauli-LDPC codes with transversal $T$ and growing distance at
vanishing rate; and a protected-subsystem construction with fully addressable
transversal $T$, constant protected rate, distance $\Omega(\sqrt n)$, and sparse (non-Pauli) checks.
More generally, we obtain asymptotically good binary codes with full transversal
addressability for every fixed finite-order one-qubit gate, retaining CSS
realizations for dyadic phase gates such as $\sqrt T$ and using generally
nonadditive codes for rotations such as $R_z(\pi/7)$.
The remaining challenge is to construct asymptotically good Pauli-LDPC codes
with fully addressable transversal non-Clifford one-qubit gates.
\end{abstract}

\section{Introduction}\label{sec:intro}

Recent constructions of constant-rate codes with transversal $T$, including
Wills's family~\cite{Wills}, motivate studying logical addressability together
with rate, distance, and check sparsity. We give a direct construction of
asymptotically good CSS codes with fully addressable transversal $T$, and
extend this to other fixed rational-angle rotations on binary codes.
Our construction produces each logical selector directly at the target phase
level. We also give two complementary sparse constructions.

A transversal operation is a product of one-qubit physical unitaries, without
a physical permutation; the \emph{projective order} of $U$ is the least
$L>0$ for which $U^L$ is a scalar. For a fixed gate of order $L$, our strongest
addressability interface is
\begin{equation}
 \left(\bigotimes_{j=1}^n U^{\epsilon_j(a)}\right)V
 =V\left(\bigotimes_{i=1}^k U^{a_i}\right),
 \qquad a\in\Z_L^k,\quad \epsilon(a)\in\{0,1\}^n,
 \label{eq:literal}
\end{equation}
where $V$ and the logical basis are fixed. Thus a physical $I/U$ mask selects
any subset of logical qubits, and also realizes arbitrary logical powers.
We additionally require the all-site layer to implement the collective target,
\begin{equation}
 U^{\otimes n}V=VU^{\otimes k}.\label{eq:strong}
\end{equation}
Gate identities are understood up to a state-independent scalar. Throughout,
asymptotic goodness means $k,d=\Omega(n)$, and LDPC means bounded check support
\emph{and} bounded participation per qubit.

The gate need not be $T$: for each fixed rotation through a rational
multiple of $\pi$ about a Pauli axis, we obtain good codes with full
transversal addressability. The strongest gate-set guarantee is for
dyadic phases $R_\ell=\diag(1,e^{2\pi i/2^\ell})$, including $R_3=T$
and $R_4=\sqrt T$: these have CSS realizations that also support every
logical Pauli transversally. Other rational angles, such as $R_z(\pi/7)$,
use generally nonadditive codes, without the same guarantee for the
other logical gates. \Cref{thm:good} gives the precise statement;
\cref{sec:gate-scope} explains the compatible gate sets and their costs.

Our three constructions are as follows.
\begin{enumerate}
\item \textbf{Good addressable codes (\cref{sec:good}).}
Explicit good binary codes satisfy \eqref{eq:literal}--\eqref{eq:strong},
with polynomial-time ideal recovery. Dyadic phase gates admit positive CSS
realizations, meaning that the defining stabilizers all have eigenvalue $+1$. A further {\em finite-clock extension} get us every fixed finite-order one-qubit gate.
\item \textbf{Pauli-LDPC codes (\cref{sec:pauli}).}
A quaternion quotient of a three-dimensional color-code geometry gives
$k=2$ and $d=\Omega(n^{1/3})$, with bounded Pauli checks. A signed
transversal implements $T\otimes T^\dagger$; a change of
stabilizer signs gives uniform physical $T$ inducing $T\otimes T$.
This construction provides sparse Pauli checks and growing distance, but with vanishing rate, and without individual addressability.
\item \textbf{Sparse protected addressability (\cref{sec:protected}).}
We construct codes with sparse commuting checks, constant protected rate, distance $\Omega(\sqrt n)$, and fully addressable transversal T on the protected subsystem.

A modular product and a fixed binary embedding give $K=\Theta(n)$ protected
qubits, distance $\Omega(\sqrt n)$, and sparse commuting projectors.
Each protected logical $T$ is individually addressable. The transversal operation also applies a known Clifford gate to the gauge qubits, without affecting the intended protected action. The binary projectors are generally non-Pauli.
\end{enumerate}

\Cref{sec:prior} gives the relevant prior constructions and the ingredients
used in each approach. \Cref{sec:future} identifies the remaining intersections
of these properties.

\section{Prior Work}\label{sec:prior}

\paragraph{Transversal phases and coding parameters.}
Divisibility and overlap constraints give binary stabilizer codes with
transversal gates throughout the dyadic hierarchy~\cite{Haah,HH}.
Wills~\cite{Wills} obtains explicit constant-rate codes with collective
transversal $T$ and growing subpolynomial distance by puncturing decreasing
monomial codes. This provides a concrete starting point for the stronger
parameter and control requirements considered here.

A distance-two code satisfying \eqref{eq:strong} already restricts the
gate's projective order. The Eastin--Knill connected-group
argument~\cite{EK} supplies the infinite-order obstruction. For rotations
about a Pauli axis, the diagonal stabilizer classification imposes the dyadic
restriction~\cite{AJO}. General diagonal-logical-operator formalisms
\cite{WQB,RK} describe the corresponding phase constraints, while
Guyot and Jaques~\cite{GJ} study limitations on CSS addressability.

\paragraph{Good codes and logical addressability.}
Algebraic-geometric methods give asymptotically good quantum codes with
transversal non-Clifford gates~\cite{GG,Nguyen}, and addressable
controlled-phase constructions~\cite{HVWZ}. Tensor-network constructions
also realize addressable non-Clifford gates~\cite{CaoLackey}.

Finite transversal groups beyond stabilizer constructions occur in
nonadditive codes. In particular, Kubischta and Teixeira~\cite{KT23}
construct distance-three codes encoding one qubit with transversal
binary-icosahedral symmetry. These illustrate the low-order exceptions
to the Pauli-axis restriction discussed in \cref{sec:gate-scope}; our
finite-clock extension instead gives good parameters and addressability
for a prescribed cyclic group.

\paragraph{Sparse constructions and operational consequences.}
Three-dimensional color codes supply Pauli-LDPC realizations of transversal
$T$ with a cubic volume cost. The unfolding theorem of Kubica, Yoshida, and
Pastawski~\cite{KYP} relates globally colored codes to toric codes and will
identify the trivial logical sector on our sphere cover. The quaternion
construction uses a quotient whose colors are permuted around loops, and
extracts its exact logical phase from an equivariant degree.

Classical good LDPC codes~\cite{Gallager} and the hypergraph-product
construction~\cite{TZ} provide the sparse ingredients for
\cref{sec:protected}. We work over $\Z_8$, prove the required free-module
and distance statements directly, and then encode each clock into a fixed
binary cell. Keeping the full clock's logical register is what permits a
sparse checked space with a protected subsystem.

The good codes also support efficient recovery and magic-state distillation.
Polynomial-time AG decoding without an additional genus
penalty~\cite{BeelenMontanucci} gives the needed linear correction radius.
Constant-overhead magic-state distillation is established in
\cite{WHY}. Our fixed-level formulation uses native
$R_\ell\ket+$ resources and ideal Clifford operations; for higher levels,
a capped injection procedure replaces perfect lower-level non-Clifford
corrections.

\paragraph{Ingredients used by the three approaches.}
\Cref{sec:good} uses the Garcia--Stichtenoth towers~\cite{GS}, their effective
AG-code constructions~\cite{Hess,Shum}, binary alphabet reduction, odd-weight
phase completion, and AG decoding~\cite{BeelenMontanucci}; its
finite-order extension uses clock CSS codes and a local charge encoding.
\Cref{sec:pauli} uses the three-dimensional color-code check complex and
the unfolding theorem~\cite{KYP}, together with the quaternion action and
degree calculation given there. \Cref{sec:protected} combines classical
LDPC input~\cite{Gallager}, the product architecture of~\cite{TZ}, and the
charge encoding from \cref{sec:good}, retaining the logical clock degrees of
freedom as a protected register and a gauge register.

\paragraph{Concurrent work.}
After the three main constructions here had been developed, but before
manuscript preparation was complete, the concurrent work of
San-Jos\'e~\cite{SanJose} appeared, with some overlapping results.
Both works obtain asymptotically good binary CSS codes with correction-free
transversal dyadic gates and addressability, as well as direct
constant-overhead $T$-state distillation; the binary dictionary and
odd-weight completion also overlap with his Proposition~4.5 and Lemma~6.2.
Our addressability proof constructs selectors by interpolation at reserved
evaluation points directly at the target phase level, whereas his
Remark~6.8 descends from a collective gate one level higher using a
logical-Pauli commutator.
Avoiding that higher-level construction, together with our four-bit boundary
packing, improves the explicit rate and relative-distance guarantees in the
comparison of \cref{app:parameters}, which uses the generic inner-map bound
for his hierarchy-lowering route.
Beyond this common CSS setting, we extend addressability to every fixed
finite-order one-qubit gate through modular clock codes, and give the
quaternion Pauli-LDPC and sparse protected constructions in
\cref{sec:pauli,sec:protected}, using geometric and ring-product methods
outside the scope of that work.
Our higher-level distillation protocol also uses only native noisy resource
states and ideal Clifford operations, without assuming ideal lower-level
non-Clifford corrections.

\section{Good codes with fully addressable finite-order gates}\label{sec:good}

We first construct good CSS codes for dyadic phase gates, with transversal
representatives of every logical Pauli. We then extend addressability to
any fixed finite-order one-qubit gate using a more general binary code.
This extension guarantees the chosen gate and its powers, without
automatically retaining transversal logical Paulis. We begin by stating
these guarantees and explaining which gates can coexist on one code.

\subsection{Supported gates, costs, and compatibility}\label{sec:gate-scope}

\begin{theorem}[Good addressable codes]\label{thm:good}
For every fixed $\ell\ge2$, there is a polynomial-time explicit family of
positive binary CSS codes with $k,d=\Omega(n)$ satisfying
\eqref{eq:literal}--\eqref{eq:strong} for
$R_\ell=\diag(1,e^{2\pi i/2^\ell})$.
Every logical Pauli has a transversal physical Pauli representative.

More generally, every fixed nonscalar one-qubit unitary $U$ of finite
projective order admits such a family of binary subspace codes, with
explicitness relative to its fixed eigenbasis. Both families have
polynomial-time computable address masks and polynomial-time ideal recovery
correcting a linear number of errors. All asymptotic constants may depend
on the fixed gate. In the general case, the asserted transversal gate set
consists of the independently addressed powers of $U$.

Conversely, a code of distance at least two with $k\ge1$ satisfying \eqref{eq:strong}
forces $U$ to have finite projective order. For a Pauli-axis rotation in an
aligned logical Pauli frame, an ordinary binary stabilizer realization also
requires that order to be a power of two.
\end{theorem}

For a Pauli $P\in\{X,Y,Z\}$, write $R_P(\theta)=e^{-i\theta P/2}$.
Its projective order is finite exactly when $\theta/\pi$ is rational.
The dyadic case has $\theta=2\pi a/2^\ell$; after cancelling common
factors, $a$ is odd. Powers of $R_\ell$ give every such $Z$ rotation.
A Clifford change of basis gives the corresponding $X$ or $Y$ rotation
on a Clifford-equivalent stabilizer code, preserving the full logical
Pauli gate set. Each fixed target has its own code family; the theorem
does not put all rational rotations on one code.

For the CSS families, an addressed phase pattern takes one layer of at
most $n$ physical gates, each either $I$ or $R_\ell$. An arbitrary logical
Pauli likewise takes one layer of physical Paulis. The rate and relative
distance are positive constants; \eqref{eq:constants} gives explicit
values for $T$. These constants can deteriorate substantially as the
gate order increases. Throughout, the polynomial-time guarantees for
construction and recovery hold for a fixed gate order as blocklength
grows. The clock extension has the same asymptotic guarantees for the
chosen phase gate, but generally loses the stabilizer structure and the
guaranteed transversal logical Paulis.

\begin{proposition}[Pauli compatibility]\label{prop:pauli-compatibility}
Suppose a binary subspace code of distance at least two has transversal
implementations of $X$, $Z$, and $U$ on one logical qubit, each acting as
identity on the remaining logical qubits. Then the projective group
$\langle X,Z,U\rangle$ is finite. If $U$ has projective order greater than
five, its rotation axis is one of the three Pauli axes.
\end{proposition}

\begin{proof}
The code-preserving transversal unitaries form a compact Lie group.
Distance at least two implies that every one-site Hamiltonian has scalar
compression to the code. A continuous evolution within this group has
a generator that is a sum of such Hamiltonians, so it acts on the code only
by an overall phase. The identity component of the group has trivial
projective logical action. A compact Lie group has finitely many
components, so the projective logical image is finite, as in the
Eastin--Knill argument~\cite{EK}.

Projective one-qubit unitaries act as rotations of the Bloch sphere.
A finite rotation group is cyclic, dihedral, tetrahedral,
octahedral, or icosahedral. The last three have no element of order greater
than five, and a cyclic group cannot contain the projective Pauli group.
The group must therefore be dihedral. Its rotations of order greater than
two share a principal axis; its other nonidentity rotations are half-turns
about perpendicular axes. The three mutually orthogonal Pauli half-turn
axes cannot all lie in that perpendicular plane, so one is the principal
axis, which is also the axis of $U$.
\end{proof}

Thus a rotation by $\pi/4$ or $\pi/8$ must be about a Pauli axis if
logical $X$ and $Z$ are also to remain transversal in the same frame.
The relevant restriction is whether the gates generate a finite group.
Irrational axis coordinates alone do not rule this out. For example,
an icosahedron with vertices
$(0,\pm1,\pm\varphi)$, $(\pm1,\pm\varphi,0)$, and
$(\pm\varphi,0,\pm1)$, where $\varphi=(1+\sqrt5)/2$, admits all three
Pauli half-turns and an order-five rotation about $(0,1,\varphi)$.
Known nonadditive distance-three codes realize transversal
binary-icosahedral symmetry~\cite{KT23}. The proposition gives a necessary
finite-group condition; it does not establish the existence of good codes
supporting each compatible group.

The general construction in \cref{thm:good} still realizes any fixed
rational-angle rotation about any axis, with constant rate, linear
distance, and full cyclic addressability. The logical frame matters:
changing the axis by a basis conjugation also conjugates the other
logical gates. It therefore does not add the new rotation to a gate set
containing the original Paulis. In particular, a non-Pauli-axis rotation
of order greater than five cannot coexist with transversal logical $X$
and $Z$ on any code of distance at least two.

For $R_z(\pi/7)$ the axis obstruction does not apply, but the dyadic
stabilizer restriction does: our generally nonadditive clock construction
provides the cyclic gate family without establishing all logical Paulis.
For the $T$ families, the transversal gates do not include the full
one-qubit Clifford group: $\langle X,Z,T\rangle$ is finite, whereas
adding $H$ makes it infinite.
An irrational multiple of $\pi$ as the rotation \emph{angle}, in contrast
to an irrational axis coordinate, has infinite projective order and is
excluded already by \cref{thm:good}.

The order restrictions in \cref{thm:good} follow from the results reviewed in
\cref{sec:prior}. We now prove the construction statements. The main task
is to choose physical phases that depend on just one logical bit and are
constant throughout its CSS coset. We obtain this control by reserving
evaluation points of an algebraic-geometric code as logical coordinates.
Interpolation expresses each reserved value using the remaining, physical
coordinates. A fixed binary encoding turns these interpolation identities
into overlap identities, and a modular correction turns them into exact
phase identities. For $T$, these two steps use fourfold overlaps and
arithmetic modulo eight. Finally, a constant amount of repetition replaces
physical gate powers by literal $I/T$ masks, including the all-site layer.
The same method works at every fixed dyadic level. The later clock
construction handles general finite order.

\subsection{Exact addressability as a phase identity}

Use the CSS convention $A\subseteq C\subseteq\F_2^{n_0}$, with stabilizers
$X(A)$ and $Z(C^\perp)$ and logical labels in $C/A$. Choose representatives
$g_1,\ldots,g_k$ of a quotient basis. The logical basis state labelled by
$b\in\F_2^k$ is the uniform superposition of the strings
$a+\sum_u b_ug_u$, as $a$ ranges over $A$. A physical $T$-power layer
with exponents $p\in\Z_8^{n_0}$ gives a string $x$ the phase
$e^{\pi i\sum_jp_jx_j/4}$. To implement $T$ on logical qubit $i$, this
must equal $e^{\pi ib_i/4}$ on every string in the superposition.
Equivalently, we require
\begin{equation}
 \sum_jp_j\left(a+\sum_u b_ug_u\right)_j=b_i\pmod8,
 \qquad a\in A,\quad b\in\F_2^k.\label{eq:phase-target}
\end{equation}
The sum inside parentheses is binary; its resulting bits are then read
as integers $0$ or $1$ in the weighted sum. Requiring the identity for
every $a$ makes the phase constant throughout each CSS coset, preserving
the logical basis state and giving it the desired phase.

Put $r=\dim A$, and let the rows $h_u$ consist of a basis of $A$ followed
by the $g_i$, so logical qubit $i$ has row index $r+i$. Expand binary
addition as an integer polynomial, starting from
$x\oplus y=x+y-2xy$. Matching the coefficients of the row-selection
bits in \eqref{eq:phase-target} gives the weighted overlap conditions
\begin{align}
 \sum_jp_jh_{u,j}&=\delta_{u,r+i}\pmod8,\nonumber\\
 \sum_jp_jh_{u,j}h_{v,j}&=0\pmod4\quad(u<v),\label{eq:ward}\\
 \sum_jp_jh_{u,j}h_{v,j}h_{w,j}&=0\pmod2\quad(u<v<w).\nonumber
\end{align}
Terms involving four or more distinct rows have coefficients divisible
by eight, so they impose no further conditions here. At level $\ell$,
the analogous conditions involve $j$-fold overlaps modulo
$2^{\ell-j+1}$ for $j\le\ell$.

Rather than solve these modular conditions directly, we first construct
a binary overlap identity that isolates the desired logical row. We
then choose physical weights that lift this parity identity to the
required phase identity. For $T$, it suffices to control fourfold
overlaps: this makes the span of pairwise codeword products
self-orthogonal and lets us correct the next two binary digits of the
weighted sum. At general level $\ell$, we use overlap order
$2^{\ell-1}$ as a sufficient condition for the same procedure.

\subsection{Logical labels at reserved evaluation points}

The classical input is an evaluation code: a function is encoded by its
values at many points. Rational functions on an algebraic curve play the
role that low-degree polynomials play in Reed--Solomon codes.
Let $\mathcal L(G)$ be the space of rational functions on a genus-$g$ curve
with no poles except at $P_\infty$, of order at most $a$; write $G=aP_\infty$.
Take $N$ other rational evaluation points over a fixed square field $\F_q$.
We use the following evaluation-code properties. A nonzero function has
at most $a$ zeros, and $\dim\mathcal L(G)=a-g+1$ when $a>2g-2$.
Values at $t$ points can be prescribed independently when $a-t>2g-2$.
The evaluation code obtained by requiring the function to vanish at
those $t$ points has dual designed distance $a-t-2g+2$ on the remaining
points. Finally, pole degrees add under multiplication. These facts
supply the interpolation and distance estimates below; the genus gives
the penalties relative to ordinary polynomial evaluation.

Reserve $t$ evaluation points $\Gamma$ for logical labels and retain the
other $N-t$ as physical symbols. Restrict the value at each reserved point
to $\F_2u=\{0,u\}$ for a nonzero $u\in\F_q$, so each reserved value
specifies one bit. Let $\mathcal U$ be the resulting binary-linear
function space and $\mathcal U_0$ its subspace vanishing on $\Gamma$.
Under the interpolation condition $a-t>2g-2$, the reserved bits are
independently assignable, giving
\begin{equation}
 \mathcal U/\mathcal U_0\cong\F_2^t.\label{eq:boundary}
\end{equation}
Two functions have the same reserved bits exactly when their difference
lies in $\mathcal U_0$. This is the quotient that will label the logical
basis. The reserved points themselves are omitted from the physical
code; interpolation on the retained points will provide access to each
of their values.

For $T$, we want a binary mask whose overlap with any four codewords has
parity equal to the product of their chosen logical bits. Higher levels
use the corresponding higher-order overlap. To obtain it from field
interpolation, we need a field-valued expression that behaves like a
binary overlap: it must be multilinear over $\F_2$ and handle repeated
inputs consistently. Its field degree must also be bounded, so we can
apply the evaluation-code estimates. Define the gate order $L$, overlap
order $s$, and field degree $E$ by
\begin{equation}
 L=2^\ell,\qquad s=2^{\ell-1},\qquad E=2^{s-1},\qquad
 P_s(x_1,\ldots,x_s)=[z^E]\prod_{j=1}^s
       \left(\sum_{r\ge0}x_j^{2^r}z^{2^r}\right).\label{eq:frob}
\end{equation}
Here $[z^E]$ extracts the coefficient of $z^E$; only finitely many terms
of the series can contribute. Each map $x\mapsto x^{2^r}$ is linear
over $\F_2$, so $P_s$ is binary multilinear. It is also symmetric and
has total field degree $E$.

The choice of coefficient ensures that the value depends only on the
set of distinct arguments, not on their multiplicities. To check this,
write $F_x$ for the series in \eqref{eq:frob}. In characteristic two,
$F_x^2+F_x=xz$. Repeatedly using this identity to eliminate repeated
factors leaves the product of the distinct factors plus correction terms $z^u$ times
$h$ factors, with $2u+h\le s$. A product of $h$ Frobenius-series terms
can contribute only at an exponent expressible as a sum of $h$ powers of two.
The correction term could contribute at degree $E$ only if $E-u$ were
a sum of $h$ powers of two. This is impossible: for $0<u<E$, the number
of ones in the binary expansion of $E-u$ is at least $s-u>h$.
In particular,
\begin{equation}
 P_s(x,\ldots,x)=x^E,\qquad
 P_s(x,\ldots,x,y)=x^{E/2}y^{E/2}.\label{eq:frob-repeat}
\end{equation}
For $T$, the selector uses $s=4$ and $E=8$.

We next encode each field element into bits while preserving the
multilinear expression above. Choose a binary basis of $\F_q$, with
$q=2^m$, and use the field trace
$\Tr(x)=\sum_{r=0}^{m-1}x^{2^r}\in\F_2$. For every nonempty
$J\subseteq[m]$ of size at most $s$, let $\lambda_J$ be the parity of
the coordinates indexed by $J$. Define the fixed inner map
\begin{equation}
 \Psi(x)=(\lambda_J(x))_J\in\F_2^B,
 \qquad B=\sum_{j=1}^{\min(m,s)}\binom mj.\label{eq:dictionary}
\end{equation}
The singleton parities retain the original binary coordinates, so $\Psi$
is injective. For any $w\in\F_q$, the tensor $\Tr(wP_s)$ is a binary
sum of $\lambda_J^{\otimes s}$. Equivalently, its value on
$x_1,\ldots,x_s$ is a sum modulo two of products
$\lambda_J(x_1)\cdots\lambda_J(x_s)$, which are overlaps of the
encoded bits. To find the coefficients, evaluate on tuples of basis
vectors. Each entry depends only on the set $I$ of distinct basis
indices, and $\lambda_J^{\otimes s}$ contributes one exactly when
$I\subseteq J$. Solving from the largest subsets downward gives a
triangular linear system. This is the subset dictionary also used
in~\cite[Proposition~4.5]{SanJose}.

Choose the reserved nonzero value $u$ so that $\Tr(u^E)=1$, and suppose
\begin{equation}
 a-t>2g-2,\qquad Ea<N-t.\label{eq:ag-conditions}
\end{equation}
The second inequality makes evaluation on the retained points injective
for $\mathcal L(EG)$: a nonzero function in this space cannot vanish
at all $N-t$ points. Its retained values therefore determine its value
at any reserved point, by a linear formula
\begin{equation}
 h(P_i)=\sum_{P\notin\Gamma}w_{iP}h(P),
 \qquad h\in\mathcal L(EG).\label{eq:interpolation}
\end{equation}
For $f_1,\ldots,f_s\in\mathcal U$, the function
$P_s(f_1,\ldots,f_s)$ lies in $\mathcal L(EG)$ because its total field
degree is $E$. Apply \eqref{eq:interpolation} and take the trace.
At the reserved point, each input is either zero or $u$; the normalization
$\Tr(u^E)=1$ makes the result exactly the product of the $i$th logical
bits. On the retained points, \eqref{eq:dictionary} expresses the trace
as a binary overlap. This gives a physical mask for that one logical
coordinate.

Let $A\subseteq C$ be the images of $\mathcal U_0\subseteq\mathcal U$
under retained evaluation and $\Psi$, and let $v_j$ be the columns of a
generator matrix $G_C$, with the $r$ stabilizer rows first. For logical qubit $i$
we have thus constructed a mask $y^{(i)}$ satisfying
\begin{equation}
 \sum_j y^{(i)}_j v_j^{\otimes s}=e_{r+i}^{\otimes s}.\label{eq:selector}
\end{equation}
An entry of this tensor identity is a statement about the overlap of
$s$ generator rows on the selected physical coordinates. Its parity is
one only when every row is the $i$th logical row. Any overlap involving
a stabilizer row or a different logical row has even parity. This is
the required isolation of logical qubit $i$.

\subsection{From overlap selectors to exact transversal gates}

The selector identity controls parities, while the gate requires an
identity modulo $L$. The next lemma supplies the physical weights that
make this conversion possible. For a binary code $D$, write $D^{*r}$
for the span of coordinatewise products of $r$ words. These spaces let
us track the overlaps needed at each successive power of two.

\begin{lemma}[Odd-weight completion; cf.~{\cite[Lemma~6.2]{SanJose}}]
\label{lem:completion}
If all $2^{\ell-1}$-fold overlaps of $D$ are even, odd coefficients
$c_j\in\Z_{2^\ell}$ can be computed such that
$\sum_jc_jx_j=0\pmod{2^\ell}$ for every $x\in D$.
\end{lemma}

For $T$, the lemma corrects two binary digits of the weighted sum.
Let $H=D^{*2}$. Fourfold evenness makes $H$ self-orthogonal, so
$x\mapsto\wt(x)/2\pmod2$ is linear on $H$. Extend this functional to
the ambient space as $z_1\cdot x$. Choosing coefficients $1+2z_1$
cancels that binary digit, making the weighted sum divisible by four
on $H$. To reach divisibility by eight on $D$, consider the remaining
digit $x\mapsto\sum_j(1+2z_{1,j})x_j/4\pmod2$. It is linear on $D$
because products of two words of $D$ lie in $H$, where divisibility by
four has already been established. Extend this second functional as
$z_2\cdot x$ and add $4z_2$ to the coefficients. At general level
$\ell$, the same argument descends through
$D^{*2^{\ell-2}},D^{*2^{\ell-3}},\ldots,D$, cancelling one carry at a time.

To apply the lemma to a logical selector, retain the columns selected by
$y^{(i)}$ and append one formal column $e_{r+i}$. By
\eqref{eq:selector}, adding this column makes the $s$-fold tensor sum
zero over $\F_2$. Thus all $s$-fold overlaps of the augmented row space
$D$ are even, as the lemma requires. It assigns an odd coefficient to
every column, including the formal one. That coefficient is invertible
modulo $L$, so we can rescale all coefficients to make it $-1$.
The formal column evaluates to the logical bit $b_i$, giving
\begin{equation}
 \sum_j\gamma^{(i)}_j(vG_C)_j=b_i\pmod L,
 \qquad v=(a,b)\in\F_2^{\dim A}\times\F_2^k.\label{eq:exact-selector}
\end{equation}
This holds for every stabilizer coefficient vector $a$ and logical label
$b$, so it is the exact phase identity needed for $R_\ell$ on logical
qubit $i$. Adding the vectors $\gamma^{(i)}$, with the desired logical
exponents as coefficients, implements any logical exponent pattern.
The formal column has served only to specify the target phase and is
not included as a physical qubit.

A constant amount of repetition replaces these physical powers by
literal $I/R_\ell$ masks. Set
$\gamma(\mathbf1)=\sum_i\gamma^{(i)}$ and repeat coordinate $j$ exactly
$r_j$ times, where
\begin{equation}
 L-1\le r_j\le2L-2,\qquad
 r_j\equiv\gamma_j(\mathbf1)\pmod L.\label{eq:repetition}
\end{equation}
Applying $R_\ell$ to $b$ copies of a repeated coordinate has the same
effect as applying $R_\ell^b$ to the original coordinate. Since
$r_j\ge L-1$, every exponent $0\le b<L$ is available. Applying the
gate to all $r_j$ copies gives exponent $\gamma_j(\mathbf1)$, so the
all-site layer implements the collective logical target. This repetition
is a CSS encoding enforced by local $ZZ$ equality checks. It preserves
$k$, does not reduce $d_Z$, and multiplies $d_X$ by at least $L-1$.

\subsection{Parameters, distance, and boundary packing}

Before repetition, $n_0=B(N-t)$, $k=t$, and
\begin{equation}
 d_X\ge N-t-a,\qquad d_Z\ge a-t-2g+2.\label{eq:ag-distance}
\end{equation}
For the $X$ bound, a nonzero outer function has at most $a$ zeros, so it
is nonzero at at least $N-t-a$ retained points. Injectivity of $\Psi$
ensures that each such point contributes a nonzero binary block.
For the $Z$ bound, a binary $Z$ pattern in one block defines a linear
functional on its field symbol. Represent that functional by a trace-dual
field element and do this for every block. Commutation with the $X$
stabilizers puts the resulting word in the field dual of the evaluation
code for $\mathcal U_0$, since $\mathcal U_0$ is field-linear. A nonzero
word in this dual has weight at least $a-t-2g+2$. If the compressed
word is zero, the original $Z$ pattern is a stabilizer. Compression
cannot increase support, proving the quantum distance bound even though
the inner code itself has short $Z$ stabilizers.

Choose a Garcia--Stichtenoth tower, a sequence of curves over one fixed
square field, with
$N/g\to\mathcal A=\sqrt q-1>2E$, and take
\begin{equation}
 \frac aN\longrightarrow\frac{1+2/\mathcal A}{E+1},\qquad
 \frac tN\longrightarrow\tau=
 \frac{1-2E/\mathcal A}{2(E+1)}>0.\label{eq:ag-parameters}
\end{equation}
These choices leave a margin asymptotic to $\tau N$ in each inequality
of \eqref{eq:ag-conditions}, and both distances in
\eqref{eq:ag-distance} grow linearly. The repetition overhead gives
$n\le(2L-2)B(N-t)$, while $k=t$ is also linear in $N$. This proves
constant rate and linear distance. The required AG computations are
effective by~\cite{Hess,Shum}. For fixed $\ell$, the field and inner
encoding have constant size; finding the selectors and correcting the
modular weights takes polynomial time in blocklength. The order-two
case follows directly from logical $Z$ representatives of good CSS codes
and the same repetition step.

The constants improve if several logical bits share one reserved point.
Replace the line $\F_2u$ by a binary subspace $W\subset\F_q$, with
trace functionals that isolate each coordinate bit in the overlap
identity. For $T$, a four-dimensional $W\subset\F_{1024}$ works with
$B=385$, so each reserved point carries four logical bits.
Here $\mathcal A=31$, $a/N\to22/186$, and $t/N\to5/186$. Packing gives
$k=4t$, pre-repetition length $n_0=385(N-t)$, and final length $n\le14n_0$.
Thus $n_0/k\to13937/4$ and $\liminf d/n_0\ge1/13937$, giving
\begin{equation}
 \liminf\frac{k}{n}\ge\frac4{195118},\qquad
 \liminf\frac d n\ge\frac1{195118}.\label{eq:constants}
\end{equation}
Packing improves the constants; the one-bit boundary already proves the
asymptotic theorem. A concrete example is given in the footnote.\footnote{Use
$\F_{1024}=\F_2[\theta]/(\theta^{10}+\theta^3+1)$, the ordered basis
$(1,603,111,518)$ of $W$, and trace weights $(453,267,799,356)$, writing
field elements as integers whose binary digits are polynomial coefficients.
The $i$th trace weight applied to $P_4(w_a,w_b,w_c,w_d)$ is one exactly
when $a=b=c=d=i$. The $4^5$ basis contractions verify the identity by
multilinearity.}

This construction has an intrinsic sparsity limitation. Every nonzero
$X$ stabilizer comes from a nonzero outer function and therefore has
linear weight. Since $\dim A=\Theta(n)$, changing the generating basis
cannot give bounded-weight $X$ checks. Thus these good CSS codes do not
become LDPC codes merely by choosing different stabilizer generators.

\subsection{Finite clocks and arbitrary rotation axes}\label{sec:clocks}

For a gate of general projective order $L$, choose an eigenbasis and remove
its overall phase to write $U=\diag(1,\zeta^u)$, with
$\zeta=e^{2\pi i/L}$ and $\gcd(u,L)=1$. We first build a good code for
finite-dimensional systems whose basis labels are residues modulo an
integer, which we call clocks. Their logical phase operators will supply
addressability, and a fixed local encoding will turn those phases into
physical one-qubit gates.

A $Q$-level clock has shift $X_Q\ket{x}=\ket{x+1\bmod Q}$ and phase
$Z_Q\ket{x}=e^{2\pi ix/Q}\ket{x}$. For matrices $A,B$ over $\Z_Q$ with
$AB^T=0$, take $X_Q$ checks from the rows of $A$ and $Z_Q$ checks from
the rows of $B$. The orthogonality condition makes them commute.
Computational labels satisfying the $Z_Q$ checks lie in $\ker B$, and
the $X_Q$ checks shift them by $\operatorname{row}(A)$. The encoded basis
therefore consists of uniform sums over the corresponding cosets.

For a prime-power clock dimension $Q=p^e$, begin with full-row-rank
prime-field CSS matrices $A_0,B_0$ satisfying $A_0B_0^T=0$, positive
quantum rate, and linear distance in both \emph{full} kernels
$\ker A_0$ and $\ker B_0$. Thus every nonzero vector in either kernel
has linear weight, including vectors representing stabilizers. This is
stronger than a quantum distance bound on logical classes alone, and
will let us preserve distance when lifting the matrices modulo $p^e$.
Such input codes can be obtained as follows. On a tower over a fixed square
$\F_{p^m}$ with $N/g\to\mathcal A>6$, take nested evaluation codes
$V_0\subset V_1$ with pole degrees $a_0\sim N/3$ and $a_1\sim2N/3$.
Let the $X$ checks generate $V_0$ and the $Z$ checks generate $V_1^\perp$.
The logical dimension is $a_1-a_0\sim N/3$, while the two full kernels
$V_0^\perp,V_1$ have distances at least $a_0-2g+2$ and $N-a_1$.
Expanding symbols in trace-dual bases gives the required prime-field
matrices, with constant length factor $m$ and the same lower bounds.

Choose any lift $A$ of $A_0$ modulo $p^e$. We adjust a lift of $B_0$
one base-$p$ digit at a time to make it commute with $A$. If $B^{(r)}$
already commutes modulo $p^r$, then $AB^{(r)T}/p^r\pmod p$ is the
next commutation error. Cancel it by solving
\begin{equation}
 A_0\Delta^T=-AB^{(r)T}/p^r\pmod p,
 \qquad B^{(r+1)}=B^{(r)}+p^r\Delta.\label{eq:hensel}
\end{equation}
A right inverse of $A_0$ supplies a solution. Adding $p^r\Delta$ leaves
all previously corrected digits unchanged and makes the checks commute
modulo $p^{r+1}$.

Minors that are invertible modulo $p$ remain invertible modulo $p^e$,
so the lifted matrices retain their ranks and their check modules have
bases over $\Z_{p^e}$. These minors also split the check maps. The
independent reductions of the $X$ rows extend to a basis of $\ker B$,
so the logical quotient $\ker B/\operatorname{row}(A)$ is free as well:
its coordinates range independently over all $p^e$ clock labels.
For distance, suppose a lifted kernel contained a short nonzero vector.
Divide out the largest common power of $p$ in its coordinates and reduce
modulo $p$. This produces a nonzero vector in the original kernel on no
more sites, contradicting its distance bound. Thus both distances remain
linear. The corrected matrices can be dense; sparsity is not required
for this construction.

For general $L$, factor it into prime powers and combine the corresponding
clocks by the Chinese remainder theorem. To match their logical dimensions,
choose suitable tower levels and fix any surplus logical coordinates.
The bounded ratio of successive tower lengths preserves goodness. The number
of prime factors and the lifting depth are constants for the fixed gate.

We now encode each clock into qubits. The binary encoding must allow
every clock phase power to be selected by an $I/U$ mask, while also
giving a prescribed phase when $U$ is applied everywhere. A unary string
$u_x=1^x0^{L-1-x}$ has exactly $x$ ones, so uniform $U$ on that string
gives phase $\zeta^{ux}$. Use several such strings to make the desired
phase coefficients available. Set $r=\lceil\log_2L\rceil$ and, for
$c\in\Z_L$, define
\begin{equation}
 E_c\ket{x}=
 \bigotimes_{h=0}^{r-1}\ket{u_{[2^hx]_L}}
 \otimes\ket{u_{[(c-(2^r-1))x]_L}}.\label{eq:chargecell}
\end{equation}
Here $[y]_L$ denotes the residue in $\{0,\ldots,L-1\}$. The first
subcell distinguishes all input labels, so $E_c$ is injective. The first
$r$ subcells supply phase coefficients $1,2,4,\ldots,2^{r-1}$. To
realize coefficient $b=\sum_hb_h2^h$, apply $U$ to every qubit of
subcell $h$ when $b_h=1$, leaving the last subcell untouched. The phase
is then $\zeta^{ubx}$, as required. Applying $U$ to all subcells instead
adds the coefficients to $c$, giving phase $\zeta^{ucx}$. Thus the final
subcell sets the collective action without removing any selectable
power. The cell uses $(r+1)(L-1)$ qubits.
For a prime-power component, first multiply its label by the CRT
idempotent: the residue that is one modulo that prime power and zero
modulo the other prime-power factors.

Choose a physical clock-$Z$ exponent vector $\Lambda_i$ representing
the phase $Z_L$ on logical clock $i$. Set
$c_j=\sum_i\Lambda_{ij}$ for physical cell $j$, with arithmetic modulo
$L$. To implement a logical exponent pattern $a$, use the digit mask
for $b_j=\sum_i a_i\Lambda_{ij}$ in that cell. The resulting logical
action is $\bigotimes_i Z_{L,i}^{ua_i}$. Applying $U$ everywhere gives
the collective target because the cell coefficients are the sum of all
logical representatives. The masks may be dense, but they add no checks.
Finally, restrict each logical clock to
$\spanof\{\ket0,\ket1\}$. The phase $Z_L^u$ on this sector is exactly
the desired qubit gate $\diag(1,\zeta^u)$. This restriction generally
produces a nonadditive binary code.

An error on $w$ binary qubits touches at most $w$ cells. Taking its matrix
elements through the cell isometries gives an operator on at most those
clock sites. Thus the outer distance also bounds the binary distance,
even if the error takes a cell outside its image. A scalar compression
on the full outer code remains scalar on any chosen logical subspace,
so restricting to two states per logical clock preserves this detection
bound. With fixed local overhead, both rate and relative distance remain
positive constants.

For a nondiagonal $U=W\diag(1,\zeta^u)W^\dagger$, use
$V_U=W^{\otimes n}V_{\rm diag}(W^\dagger)^{\otimes k}$.
This returns the gate identity to the original physical and logical
bases. The physical basis change is a product of one-qubit unitaries,
so it preserves distance and transversality.
For $R_z(\pi/7)$, the odd factor in projective order fourteen calls for
the clock route. The dyadic restriction excludes an aligned binary
Pauli-stabilizer implementation of distance at least two.

For completeness, infinite projective order is excluded already by
\eqref{eq:strong}. Projectively, the powers of $U$ are dense in a circle.
By continuity, the gate identity extends to this one-parameter group.
Differentiating it equates the compressed physical generator with the
logical generator. The physical generator is a sum of one-site
Hamiltonians, whose compressions are scalar at distance at least two;
the logical generator is nonscalar. This contradiction is the relevant
Eastin--Knill argument~\cite{EK}.

\subsection{Recovery and gate use}

Recovery can be performed on the larger outer code before returning to
the chosen logical subspace. For the dyadic construction, this parent
is a field-valued CSS code with the two classical decoding distances in
\eqref{eq:ag-distance}. The inner map $\Psi$ and repetition are
binary-linear isometries, so they have Clifford implementations.
Undoing them within each block turns an error on $e$ binary sites into
an error on at most $e$ outer field symbols. AG decoding
through half the designed distance, without an extra genus penalty
\cite{BeelenMontanucci}, gives polynomial-time recovery of every error on
\begin{equation}
 \left\lfloor\frac{\min\{N-t-a,\ a-t-2g+2\}-1}{2}\right\rfloor
 \label{eq:recovery-radius}
\end{equation}
binary sites. Since this recovery corrects the full parent code, it also
preserves every state and superposition in the selected logical sector.

For the clock construction, decode the error one base-$p$ digit at a
time. A lifted check matrix $A$ over $\Z_{p^e}$ gives syndrome
$s=A\boldsymbol e$ for an error vector $\boldsymbol e$. Modulo $p$,
this is the syndrome of its lowest digit under $A_0$, so the original
classical decoder recovers that digit
$\boldsymbol e_0$. Subtract $A\boldsymbol e_0$ from $s$ and divide by $p$
to expose the next digit's syndrome. Repeat for the fixed number of
digits and for both CSS error types. Subtracting a recovered digit and
dividing by $p$ introduces no errors at previously unaffected sites.
Thus every residual digit stays within the original support, and the
prime-field decoder's correction radius applies at every stage. Clock
Weyl operators span all errors on those sites, so the resulting recovery
also corrects arbitrary errors on that support.
For each binary cell, use a fixed local decoding channel that inverts
$E_c$ on its image and outputs a fixed clock state on the orthogonal
complement. An error on $w$ qubits then becomes an error channel on at
most $w$ clocks. Follow these local channels by outer recovery and
re-encoding. This gives polynomial-time recovery of the full clock parent,
preserving the restricted logical sector and all its coherences.

The dyadic codes also distill noisy resource states for their target gate.
With ideal stabilizer operations, $O(n)$ independent sufficiently accurate
copies of $R_\ell\ket+$ produce $\Theta(n)$ outputs whose joint trace
distance from the ideal product is $e^{-\Omega(n)}$. Prepare encoded
$\ket+^{\otimes k}$ with Clifford operations, inject the physical phase
layer using the noisy resources, recover the code, and decode. Without
faults the exact transversal identity gives the desired product of logical
resource states. At $T$, use the
weighted code before repetition: an odd physical exponent consumes one
raw $T$ state, with Clifford corrections, while even exponents require
only Clifford gates. To put the noise in a form the code can correct,
randomly apply either $I$ or the Clifford $TXT^\dagger$ to each raw
resource. This converts it into a mixture of $T\ket+$ and $ZT\ket+$
without changing its infidelity. Injection then produces independent
stochastic phase faults.
The four-bit packing yields an asymptotic raw-state bound $13937/4$ per
output and sufficient raw infidelity $p<1/27874$ under this independent
noise model.

For higher levels, an injection using the target gate's resource state
gives either $R_\ell$ or $R_\ell^{-1}$. Track the accumulated exponent
and repeat until it agrees with the target modulo $2^{\ell-2}$.
The remaining difference can be corrected by a power of $S$.
Because the level is fixed, a constant limit on the number of attempts
can make the timeout probability a sufficiently small constant.
Treat timeouts and resource imperfections as local channel faults.
Expanding the independent channels according to their faulty sites,
exact recovery corrects every term within its correction radius. The
remaining binomial tail is exponentially small at sufficiently low fault
rate. This uses only resources for the target gate and ideal Clifford
operations, with constant raw-state overhead; it does not assert a
threshold for noisy encoding, recovery, or other Clifford circuits.

Finally, bitwise CNOT between identical CSS blocks implements CNOT on
corresponding logical positions. We can therefore compute XOR parities
of those logical bits, apply addressed phases, and uncompute the
parities. The identity
\begin{equation}
 \sum_{\emptyset\ne J\subseteq[r]}(-1)^{|J|+1}
       \bigoplus_{i\in J}x_i=2^{r-1}\prod_{i=1}^r x_i.
 \label{eq:controlled}
\end{equation}
expresses a controlled phase in terms of these parity phases. For
$r=\ell=3$, this gives selected CCZ gates between corresponding logical
positions of three blocks.

\section{Pauli-LDPC codes from quaternion color holonomy}\label{sec:pauli}

In our second construction, we get sparse Pauli checks by changing the source
of the logical phase. Instead of selecting logical coordinates by
interpolation, we arrange that transporting four local colors around a
closed loop can permute them. This \emph{color holonomy} creates two logical
qubits on a closed three-dimensional space and makes their signed physical
weight an integer degree. The resulting phase is precisely a coupled $T$ gate, specifically a logical $T \otimes T^\dagger$.
The downside here is vanishing rate: refining it further increases distance but leaves the
number of logical qubits fixed.

\begin{theorem}[Quaternion Pauli-LDPC family]\label{thm:quaternion}
For every odd integer $M\ge3$, there is an explicit positive CSS code with
\begin{equation}
 n=2M^3,\qquad k=2,\qquad
 d_X,d_Z\ge\left\lceil\frac{\pi M}{16}\right\rceil.
 \label{eq:quaternion-parameters}
\end{equation}
Its displayed $X$ and $Z$ checks have weights at most $24$ and $6$,
respectively. Every qubit participates in four displayed $X$ checks and six
displayed $Z$ checks. There are signs $\sigma_t\in\{1,-1\}$ and a fixed
logical Pauli basis in which
\begin{equation}
 U_\sigma V=V(T\otimes T^\dagger),\qquad
 U_\sigma=\bigotimes_t T^{\sigma_t}.
 \label{eq:quaternion-gate}
\end{equation}
Changing stabilizer eigenvalue signs gives an affine CSS code with the
same parameters on which uniform physical $T$ implements logical
$T\otimes T$, up to a scalar.
\end{theorem}
The proof follows below. \cref{sec:ldpcConstruction} gives the basic construction, and \cref{sec:dbound} proves the required distance for \cref{eq:quaternion-parameters}. \cref{sec:ldpcPhaseGate} gives the $T$ construction for \cref{eq:quaternion-gate}. Finally, \cref{sec:uniformT} shows that we can get uniform physical $T$ to implement logical $T \otimes T$.

\subsection{Construction: Local Checks on a Twisted Coloring}\label{sec:ldpcConstruction}
Identify the unit three-sphere with unit quaternions. The eight elements
$Q_8=\{\pm1,\pm i,\pm j,\pm k\}$ act by left multiplication, without
fixed points. We build a colored triangulation upstairs and identify cells
related by this action. Start with the boundary of the four-dimensional
cross-polytope of some radius $M$, which we can describe as $|w| + |x| + |y| + |z| = M$. Its sixteen facets correspond to choices of signs on the
axes $(1,i,j,k)$. On each radius-$M$ facet, write its nonnegative coordinates
as $a_0+\cdots+a_3=M$. The cumulative coordinates
$(a_3,a_2+a_3,a_1+a_2+a_3)$ identify it with
$0\le x_1\le x_2\le x_3\le M$. Subdivide into lattice tetrahedra by the
unit-coordinate paths: their vertices are $x$, $x+e_{\pi(1)}$,
$x+e_{\pi(1)}+e_{\pi(2)}$, and $x+(1,1,1)$, for permutations $\pi$,
retaining the tetrahedra in the chamber. The barycentric coordinate
permutations induced by $i,j$ preserve this subdivision. The subdivisions
agree on common faces and are preserved by $Q_8$; each facet contains
$M^3$ tetrahedra.

At a lattice vertex $v$, assign the color
\[
 c(v)=M\sum_{r=0}^3 r|v_r|\pmod4.
\]
Each tetrahedron has all four colors. Identifying colors with two-bit
labels, multiplication by $i$ permutes them by $c\mapsto c\mathbin{\mathrm{xor}}1$,
and multiplication by $j$ by $c\mapsto c\mathbin{\mathrm{xor}}2$.
Thus a global coloring need not survive the quotient, although each local
neighborhood retains the color-code incidence rules.

Place a qubit on each tetrahedron orbit, an $X$ check on each vertex star,
and a $Z$ check on each edge star. There are $16M^3/8$ qubits. Counting the
unit-coordinate paths incident on a vertex gives star weights $8,16,24$;
edge stars have weights $4,6$. The four vertices and six edges of a
tetrahedron give the participation bounds. The checks commute: a vertex
and edge star overlap either in an even link cycle, in the two tetrahedra
sharing a triangle, or not at all. Define $\sigma_t$ as the determinant
sign of the four vertices in (some fixed) color order. Quaternion multiplication preserves orientation
and permutes colors evenly, so this sign descends to the quotient.

\par
We want to count the logical qubits encoded in this structure.
The globally colored sphere code encodes no logical qubits: color-code unfolding maps it to three toric codes on the sphere, each of which encodes none~\cite{KYP}. We use unfolding
only upstairs, where the coloring is global. Consequently, every upstairs
binary word satisfying the $Z$ checks is the tetrahedron parity of a
vertex potential $p$, namely $x_t=\sum_{v\in t}p(v)$ over $\F_2$.
Two potentials give the same word precisely when their difference assigns
one value to each color, with even total parity. In color order $(0,1,2,3)$, denote this ambiguity by
\[
 W=\{w\in\F_2^4:\textstyle\sum_c w_c=0\}
   =\spanof\{u=(1,1,0,0),\ v=(1,0,1,0),\ h=(1,1,1,1)\}.
\]

A quotient word lifts to an invariant upstairs word, but its potential
can change under $g\in Q_8$ by an element $z_g\in W$. Consistency requires
$z_{gq}=z_g+g z_q$: these are the elementary cocycle equations. Changing
the potential by a color assignment changes $z_g$ by $gw-w$; invariant
potentials give quotient $X$ stabilizers. Conversely, freeness on vertices
lets us choose a potential on one representative of each orbit and extend
it using any such cocycle. Thus logical words are exactly cocycles modulo
these changes of potential.

The color permutations obey
$(i-1)u=(j-1)v=0$ and $(i-1)v=(j-1)u=h$.
Substitution into the quaternion relations gives
\begin{equation}
 z_i=au+bv+ch,\qquad z_j=bu+(a+b)v+dh.
 \label{eq:quaternion-cocycles}
\end{equation}
The coefficients $c$ and $d$ can be changed freely by $gw-w$, but $a$ and $b$
can't be. Hence $C/A\cong\F_2^2$, and we can represent it with logical coordinates $(a,b)$.

\subsection{Distance Bounds}\label{sec:dbound}
We prove the two distance bounds by the same argument. A nontrivial
logical operator must have a support containing a closed qubit--check
walk whose lift to the sphere ends at a different copy of its starting
point. We will show that such a path must pass through at least
$\pi M/16$ distinct qubits.

First record the geometric estimates. Rescale the cross-polytope to
radius one and project its boundary radially onto the unit sphere.
Each small tetrahedron has Euclidean diameter at most $2/M$ before
projection. A point $v$ on a facet satisfies
$1=\sum_r|v_r|\le2\|v\|_2$, so its norm is at least $1/2$.
Radial projection therefore expands path lengths within a facet by
at most two. Represent a qubit by the projected barycenter of its
tetrahedron, an $X$ check by its projected vertex, and a $Z$ check by the projected
midpoint of its edge. Connecting these points inside each incident
tetrahedron and then projecting gives a path of length at most $4/M$
for every qubit--check incidence. On the other hand, for every $p\in S^3$
and $g\in Q_8\setminus\{1\}$, the spherical distance from $p$ to $gp$
is at least $\pi/2$.

These estimates also justify lifting the Tanner graphs. For $M\ge3$,
the spherical diameter bound $4/M<\pi/2$ prevents two vertices of one
tetrahedron from being identified by the quaternion action. Thus no
incidences at a qubit, vertex check, or edge check are merged by the
quotient: the upstairs Tanner graphs are graph covers of the downstairs
ones.\footnote{For small $M$, distinct quotient cells can have the same
vertex set. They are retained as distinct cells when forming these
incidence graphs.} In particular, choosing an upstairs copy of a graph
node determines a unique lift of any path starting there.

Now let $P$ be a nontrivial $X$-type or $Z$-type logical operator.
Form its support graph by retaining the supported qubits, all adjacent
checks of the opposite Pauli type, and all incidences between them.
Commutation with the checks means that every check node has even
degree. Each connected component satisfies these parity conditions
separately, since any check touching two supported qubits connects
them within the same component.

Suppose every closed walk in a component returns to its starting copy
when lifted upstairs. Choose one upstairs copy of one node and extend
that choice along paths. The closed-walk assumption makes the result
independent of the chosen paths, giving a lift with exactly one copy
of each node and edge in the component. The lifted support still meets
each check in an even number of qubits, so its Pauli operator commutes
with all upstairs stabilizers. Because the sphere code encodes no
logical qubits, this operator is itself a stabilizer. Projecting supports
downstairs, with multiplicities summed modulo two, maps each upstairs
stabilizer generator to its corresponding downstairs generator.
It therefore maps the lifted stabilizer to a downstairs stabilizer.
Since the lift contains only one copy of each supported qubit, its
projection is exactly the original component.

If every component had this property, $P$ would be a stabilizer,
contrary to its choice. Hence the support graph contains a closed walk
whose lift changes the starting copy, a property called nontrivial
\emph{monodromy}. Removing backtracking and decomposing closed walks
into simple cycles shows that at least one simple cycle has this
property. Let it contain $m$
qubit nodes. They are distinct because the cycle is simple, so
$\wt(P)\ge m$.

The cycle alternates between qubits and checks and therefore has $2m$
incidence edges. Its geometric lift is a path from some $p$ to $gp$ with
$g\ne1$. Combining the lower bound on endpoint separation with the
upper bound on each incidence length gives
\[
 \frac{\pi}{2}
 \le \operatorname{length}(\text{lifted path})
 \le 2m\,\frac4M=\frac{8m}{M}.
\]
Consequently $\wt(P)\ge m\ge\pi M/16$. The argument applies to either
Pauli type, and weights are integers, proving
\[
 d_X,d_Z\ge\left\lceil\frac{\pi M}{16}\right\rceil
\]
as claimed in \cref{thm:quaternion}.

\subsection{Implementing the Phase Gate}\label{sec:ldpcPhaseGate}
We now show that the signed transversal $T$ operation preserves the code and acts as $T \otimes T^\dagger$ on its two logical qubits.
\[
U_\sigma|x\rangle
=e^{i\pi f(x)/4}|x\rangle,
\qquad f(x)=\sum_t\sigma_t x_t.
\]
We must show that this phase is constant within each logical basis state’s superposition, and determine its value for the four logical labels. We start by computing the signed weight in each logical
sector modulo eight without choosing a large physical representative.
Given a potential, send each vertex to the signed color axis
$F(v)=(-1)^{p(v)}e_{c(v)}$ of a target cross-polytope. Extend over
tetrahedra to obtain a map between three-spheres. Its degree counts,
with orientation signs, how many times the domain covers the target.
The cocycle determines the target action
\[
 \rho_z(g)e_c=(-1)^{z_g(g\cdot c)}e_{g\cdot c},
\]
and $F$ respects the two $Q_8$ actions. Each tetrahedron maps to the target
facet selected by its four vertex signs, and its orientation contribution
is $\sigma_t(-1)^{x_t}$. Counting over the sixteen target
facets yields
\[
 \sum_{t\text{ upstairs}}\sigma_t(-1)^{x_t}=16\deg F,
 \qquad \sum_{t\text{ upstairs}}\sigma_t=0.
\]
Since the cover has eight sheets, the signed physical weight downstairs is
\begin{equation}
 f(x):=\sum_{t/Q_8}\sigma_t x_t=-\deg F.
 \label{eq:quaternion-degree}
\end{equation}

For a fixed target action, any two equivariant maps have degrees congruent
modulo eight. The target $S^3$ is connected and has vanishing first and
second homotopy groups, so the maps can be deformed into agreement through
dimension two, choosing the deformation on one cell per free orbit.
Remaining degree differences occur in orbits of eight three-cells.
Changing the cocycle by $gw-w$ conjugates its target action by
$\diag((-1)^{w_c})$, an orientation-preserving map because $w$ has even
weight. Thus $f(x)\bmod8$ depends only on the logical class.

Set $c=d=0$ in \eqref{eq:quaternion-cocycles}. To check the four degrees,
write $I(x)=ix$ for left multiplication on unit quaternions and
$\mathcal C(x)=\bar x$ for conjugation. The following maps intertwine the
domain action with $\rho_z$; this can be checked on $i,j$:
\[
 \begin{array}{c|c|r}
 (a,b)&\text{comparison map}&\text{degree}\\ \hline
 (0,0)&x\mapsto(1,1,1,1)/2&0\\
 (1,0)&x\mapsto(1,-1,-1,-1)/2&0\\
 (0,1)&I&+1\\
 (1,1)&\mathcal C I&-1
 \end{array}
\]
The constant maps land at fixed points; $I$ preserves orientation, while
$\mathcal C$ reverses three quaternion coordinates. Consequently,
\begin{equation}
 \begin{array}{c|rrrr}
 (a,b)&(0,0)&(1,0)&(0,1)&(1,1)\\ \hline
 f(a,b)\pmod8&0&0&-1&1
 \end{array}
 \qquad f(a,b)=a-(a\oplus b)\pmod8.
 \label{eq:quaternion-phases}
\end{equation}
In the CNOT-related coordinates $r=a$, $s=a\oplus b$, this is exactly
$T_rT_s^\dagger$, with no residual controlled phase.

\subsection{Implementing via Uniform physical \texorpdfstring{$T$}{T}}\label{sec:uniformT}
For the uniform gate, let $D$ apply Pauli $X$ at negative-sign sites.
There are $M^3$ such sites. With $\omega=e^{\pi i/4}$,
$XTX=\omega T^\dagger$ gives
\[
 DT^{\otimes n}D=\omega^{M^3}U_\sigma,
 \qquad V'=DV(I\otimes X),\qquad
 T^{\otimes n}V'=\omega^{M^3-1}V'(T\otimes T).
\]
Conjugation changes only $Z$-check eigenvalue signs, preserving locality
and distance.

This supplies the collective interface~\eqref{eq:strong}, but not
individual addressability~\eqref{eq:literal}. For $M\ge7$, every preserving
physical $T$-power layer has either all even or all odd exponents. The three
$X$ stars at a triangle meet in its two adjacent tetrahedra, so their
mod-two overlap condition forces the two exponents to have equal parity.
There are no extra triangle incidences: two distinct quotient edges with
the same endpoints would give a deck loop of length at most $8/M<\pi/2$,
so edges, and then lifted triangles, are determined by their vertices.
Connectivity propagates the parity to every tetrahedron. The logical action
is therefore Clifford or differs from
$T_rT_s^\dagger$ by a Clifford. Neither $T_r$ nor $T_s$ alone is possible
through such a layer. The construction establishes genuine Pauli-LDPC
checks and $d=\Omega(n^{1/3})$ with $k=2$; obtaining extensive logical
storage and addressability requires an additional mechanism.

\section{Sparse commuting checks with protected addressability}\label{sec:protected}

Our third construction starts with a sparse code on eight-level systems
and encodes each physical eight-level system into 28 qubits. Each logical
eight-level system carries three logical qubits. We designate one as
protected and allow a known Clifford operation on the other two whenever
we apply its logical $T$. Keeping all three logical qubits is what lets
us retain the sparse checks of the original code.

To see the gate action, label an eight-level basis by $x\in\Z_8$ and
write $x=a+2b$, with $a\in\{0,1\}$ and $b\in\{0,1,2,3\}$. This
identifies $\ket{x}$ with $\ket a_L\ket b_G$. The clock phase
$Z_8\ket{x}=e^{\pi ix/4}\ket{x}$ then factors as
\begin{equation}
 Z_8\ket{a+2b}=e^{\pi ia/4}e^{\pi ib/2}\ket{a+2b},
 \qquad Z_8=T_L\otimes Z_{4,G}.\label{eq:clock-split}
\end{equation}
Writing $b=b_0+2b_1$ shows that $Z_4=S\otimes Z$ on the two gauge
qubits. Thus an individually addressable logical clock phase gives an
individually addressable protected $T$, with a separate Clifford action
on the gauge register. The factorization holds for every gauge state,
including states entangled with the protected register.

\Needspace{18\baselineskip}
\begin{theorem}[Sparse protected transversal $T$]\label{thm:protected}
Let $R_0$ be a full-row-rank binary $m\times n$ matrix, $0<m<n$, with row
weight at most $w$, column weight at most $q$, and kernel distance $d_0$.
There is a binary commuting-projector code with
\begin{equation}
 N=28(n^2+m^2),\qquad K=(n-m)^2,\qquad d=d_0,\label{eq:protected-parameters}
\end{equation}
encoding $K$ protected qubits and $2K$ gauge qubits. Check support is at
most $28(w+q)$ and participation at most $1+2\max(w,q)$ per physical qubit.
For every $t\in\Z_8^K$, a computable mask $\mu(t)\in\{0,1\}^N$ satisfies
\begin{equation}
 T^{\mu(t)}V=V\left[
     \bigotimes_iT_{L_i}^{t_i}\otimes\bigotimes_iZ_{4,G_i}^{t_i}
     \right],
 \qquad Z_4=\diag(1,i,-1,-i).\label{eq:protected-action}
\end{equation}
Here $T^{\mu(t)}=\bigotimes_jT^{\mu_j(t)}$. Uniform physical $T$ also
implements the target $t=\mathbf1$. The full checked subspace, as well
as the protected subsystem, has exact distance $d_0$.
\end{theorem}

The full checked space encodes $3K$ ordinary qubits. Calling $2K$ of them
``gauge'' specifies which logical factor may undergo the accompanying
Clifford operation; it does not introduce a Pauli gauge group. In fact,
every operator $F$ on fewer than $d_0$ physical qubits satisfies
$V^\dagger FV=c_FI_{LG}$, so these errors are detected on the entire
encoded space. The binary check projectors are generally non-Pauli.

For good classical $(3,6)$ LDPC input, choose a maximal independent subset
of the check rows. This preserves the binary kernel and its distance,
and gives $m\le n/2$, $w\le6$, and $q\le3$. Hence $K\ge n^2/4$
and $N\le35n^2$, yielding
\begin{equation}
 \frac KN\ge\frac1{140},\qquad d=\Omega(\sqrt N),\qquad
 \text{check support}\le252,\quad\text{participation}\le13.
 \label{eq:protected-numbers}
\end{equation}
The rate counts only the protected qubits. We prove the theorem by first
constructing the eight-level code and determining its logical systems
and distance, then encoding its physical systems into qubits while
preserving the required gates and check sparsity.

\subsection{Constructing the sparse eight-level code}\label{sec:protected-parent}

Use the same $0,1$ entries of $R_0$ to form a matrix $R$ over $\Z_8$.
Its row and column supports are unchanged; only the arithmetic is now
modulo eight. Since $\Z_8$ is not a field, we must check that the logical
labels are independent coordinates ranging over all eight residues.
A rank count alone would not establish this.

Full row rank of $R_0$ gives an $m\times m$ minor of $R$ with odd
determinant, which is invertible modulo eight. Thus $R$ has a right
inverse $J$, with $RJ=I_m$. Every vector $v$ decomposes uniquely into
$v-JRv\in\ker R$ and $JRv\in J(\Z_8^m)$, giving
\begin{equation}
 \Z_8^n=H\oplus J(\Z_8^m),\qquad H=\ker R\cong\Z_8^{n-m}.
 \label{eq:ring-splitting}
\end{equation}
The last isomorphism can also be seen by solving for the $m$ coordinates
of the invertible minor: the other $n-m$ coordinates are arbitrary
elements of $\Z_8$.

We use $n^2+m^2$ physical eight-level systems, with their computational
labels arranged as matrices $(B,D)$ of sizes $n\times n$ and $m\times m$.
On each system let $X_8\ket{x}=\ket{x+1\bmod8}$ and
$Z_8\ket{x}=e^{\pi ix/4}\ket{x}$. The checks are specified by two
linear maps:
\begin{equation}
 \operatorname{Mat}_{n\times m}(\Z_8)
 \xrightarrow{\ d_2\ }\operatorname{Mat}_{n\times n}(\Z_8)
       \oplus\operatorname{Mat}_{m\times m}(\Z_8)
 \xrightarrow{\ d_1\ }\operatorname{Mat}_{m\times n}(\Z_8),
 \label{eq:ring-complex}
\end{equation}
where $d_2(A)=(AR,RA)$ and $d_1(B,D)=RB-DR$. For each single-entry
matrix $A$, the entries of $d_2(A)$ specify the exponents of an $X_8$
check. Each coordinate of $d_1(B,D)$ specifies a $Z_8$ check. The
checks commute because
$d_1d_2(A)=RAR-RAR=0$ modulo eight. Each check involves at most one
row and one column of $R$, hence at most $w+q$ physical systems.
Each physical system occurs in at most $\max(w,q)$ checks of either
type.

We next determine the encoded basis. A label $(B,D)$ satisfies the
$Z_8$ checks exactly when $RB=DR$. The $X_8$ checks shift labels by
elements of $\operatorname{im}d_2$. Consequently the encoded basis
consists of uniform sums over cosets in
$\mathcal C/\mathcal A$, where
$\mathcal C=\ker d_1$ and $\mathcal A=\operatorname{im}d_2$.
This is the same coset description as for binary CSS codes, with
addition modulo eight.

There is a concrete way to identify these cosets. If $RB=DR$ and
$v\in H$, then $R(Bv)=D(Rv)=0$, so $B$ maps $H$ into itself.
Adding $(AR,RA)$ does not change this map, since $ARv=0$ on $H$.
Conversely, every $\Z_8$-linear map $F:H\to H$ occurs: take
$Bv=F(v-JRv)$ and $D=0$. Finally, if a pair satisfying $RB=DR$
acts as zero on $H$, then
$B=BJR$ and $D=RBJ$, so $(B,D)=d_2(BJ)$ and its coset is trivial.
We have therefore proved
\[
 \mathcal C/\mathcal A\cong\End_{\Z_8}(H)
 \cong\operatorname{Mat}_{(n-m)\times(n-m)}(\Z_8).
\]
Here $\End_{\Z_8}(H)$ means all linear maps from $H$ to itself.
Each of the $(n-m)^2$ matrix entries is an independent eight-valued
logical label. Thus the code encodes exactly $K=(n-m)^2$ full
eight-level systems.

\subsection{Distance of the eight-level code}\label{sec:protected-clock-distance}

We first relate $H=\ker R$ to the original binary code. Its minimum
nonzero support is exactly $d_0$. For the lower bound, write a nonzero
$h\in H$ as $h=2^s u$, taking $2^s$ to be the largest power of two
dividing all its coordinates. Then $s\in\{0,1,2\}$ and at least one
coordinate of $u$ is odd. The equation $Rh=0\pmod8$ implies
$R_0(u\bmod2)=0$. This gives a nonzero binary codeword supported
within $\supp h$, so $\wt(h)\ge d_0$. For the reverse inequality,
if $c$ is a minimum-weight binary codeword, then $4c\in H$ and
$\wt(4c)=d_0$. We will also use that every binary codeword has a lift
in $H$: regarding $c$ as a $0,1$ vector, $c-JRc$ lies in $H$ and
reduces to $c$ modulo two because $Rc$ is even.

An $X_8$-type operator commutes with the $Z_8$ checks precisely when
its exponent pair $(B,D)$ belongs to $\mathcal C$. Suppose its total
weight, counting nonzero entries in both matrices, is less than $d_0$.
Then $B$ has fewer than $d_0$ nonzero rows. For any $v\in H$, the
vector $Bv$ belongs to $H$ and is supported on those rows, so the
kernel-distance bound forces $Bv=0$. The preceding identification of
logical cosets now shows that $(B,D)\in\mathcal A$: the operator is
an $X_8$ stabilizer. This proves the $X$-distance lower bound.

To obtain the $Z$ bound, we identify its exponent classes with the
$X$ classes by transposing the matrices. Define
\begin{equation}
 \langle(B,D),(B',D')\rangle=\operatorname{Tr}(BB')-
       \operatorname{Tr}(DD').\label{eq:ring-pairing}
\end{equation}
This pairing is nondegenerate over $\Z_8$, and
\[
 \langle d_2(A),(B,D)\rangle
 =\operatorname{Tr}\bigl(A(RB-DR)\bigr).
\]
Thus a pair is orthogonal to every $X_8$ stabilizer in this pairing
exactly when it lies in $\mathcal C$. The map
$\tau(B,D)=(B^T,-D^T)$ converts the displayed pairing into the
ordinary coordinate dot product. Writing $\perp$ for orthogonality
under that dot product, we obtain
$\tau(\mathcal C)=\mathcal A^\perp$ and
$\tau(\mathcal A)=\mathcal C^\perp$. These are respectively the
$Z_8$ operators commuting with all $X_8$ checks and the $Z_8$
stabilizers. Hence $\tau$ identifies the two logical-operator
quotients and preserves weight, proving the same $Z$-distance bound.

The bound is attained. Choose a binary word $c$ of weight $d_0$ in
$\ker R_0$ and a coordinate $j$ with $c_j=1$. The pair
$(4ce_j^T,0)$ lies in $\mathcal C$ and has weight $d_0$.
Choose a lift $v\in H$ of $c$; then $v_j$ is odd, so
$4ce_j^Tv=4c\ne0$. The pair therefore acts nontrivially on $H$
and represents a nontrivial logical operator of $X_8$ type. Applying
$\tau$ gives one of $Z_8$ type with the same weight. Both distances are exactly
$d_0$. Products of $X_8$ and $Z_8$ powers span all operators on the
physical eight-level systems, so the code detects arbitrary errors
supported on fewer than $d_0$ such systems.

\subsection{Encoding into qubits and implementing the gates}\label{sec:protected-binary-gates}

We now replace each physical eight-level system by 28 qubits. The local
encoding must allow a chosen $Z_8$ power to be implemented using only
physical $I$ and $T$. It must also allow the all-$T$ layer to implement
a separately prescribed power. Both requirements are met by the
$L=8$ cell of \cref{sec:clocks}, which we describe explicitly here.

Let $u_s=1^s0^{7-s}$ for $0\le s\le7$, and write $[y]_8$ for the
residue in $\{0,\ldots,7\}$. For $c\in\Z_8$, encode
\begin{equation}
 E_c\ket{x}=\ket{u_{[x]_8}}\ket{u_{[2x]_8}}
              \ket{u_{[4x]_8}}\ket{u_{[(c-7)x]_8}}.
 \label{eq:cell28}
\end{equation}
These are four blocks of seven qubits. The first block distinguishes
all eight input labels, so $E_c$ is an isometry. Applying $T$ to all
seven qubits of a block multiplies its state by $e^{\pi is/4}$ when
its label is $u_s$. Thus the first three blocks implement phase
coefficients $1,2,4$. To realize $Z_8^b$, write
$b=b_0+2b_1+4b_2$ and apply $T$ to the corresponding selected blocks,
leaving the fourth untouched. Applying $T$ to all four blocks instead
gives coefficient $1+2+4+(c-7)=c$ modulo eight, and hence realizes
$Z_8^c$.

Choose a basis of the $K$ logical eight-level systems. For each $i$,
let $\Lambda_i$ be a physical exponent vector such that
$\bigotimes_jZ_{8,j}^{\Lambda_{ij}}$ implements its logical $Z_8$.
To compute these representatives, choose a basis of $H$ and extend its
coordinate functionals to $\Z_8^n$ using \eqref{eq:ring-splitting}.
Each matrix entry of $B|_H$ is then a linear combination of physical
entries of $B$; its coefficients give the corresponding $\Lambda_i$.
Fix the encoding of physical system $j$ by choosing
$c_j=\sum_i\Lambda_{ij}\pmod8$. For a logical target
$t\in\Z_8^K$, compute
$b_j(t)=\sum_i t_i\Lambda_{ij}\pmod8$ and use the first three
blocks of cell $j$ to implement $Z_8^{b_j(t)}$. The resulting physical
$I/T$ mask implements exactly the requested logical clock phases.
Applying $T$ to every physical qubit instead implements the sum of
the logical phase representatives, which is the collective target.
Splitting each logical clock as in \eqref{eq:clock-split} proves the
protected and gauge actions in \eqref{eq:protected-action}.

The representatives $\Lambda_i$ can have large support. Even then,
the selected physical gates form a single transversal layer. Check
sparsity is determined by $R$ and the local encoding, independently
of the supports of these gate masks.

\subsection{Sparse binary checks and exact distance}\label{sec:protected-binary-checks}

Two types of checks define the binary code. First, include the projector
$P_j=E_{c_j}E_{c_j}^\dagger$ onto the eight-dimensional valid subspace
of each 28-qubit cell. Second, let $\Pi_h$ be the $+1$ eigenspace
projector of an outer clock check, supported on a set $A_h$ of
physical clocks. Replace it by
\[
 Q_h=E_{A_h}\Pi_hE_{A_h}^\dagger,
 \qquad E_{A_h}=\bigotimes_{j\in A_h}E_{c_j},
\]
tensored with identity outside those cells. Each $Q_h$ acts on at
most $28(w+q)$ qubits. A physical qubit belongs to one cell-image
check and at most $2\max(w,q)$ outer checks, giving the claimed
participation bound.

These projectors commute on the entire binary Hilbert space, including
states outside the valid cell subspaces. Each $Q_h$ commutes with
every $P_j$: if $j\in A_h$, then $P_jQ_h=Q_hP_j=Q_h$, and
otherwise their supports are disjoint. We can therefore examine the
sectors in which each cell is valid or invalid separately. If a cell
in $A_h\cup A_{h'}$ is invalid, at least one of $Q_h,Q_{h'}$
annihilates that sector, and both product orders vanish. If all cells
in the union are valid, the two products are the encoded versions of
$\Pi_h\Pi_{h'}$ and $\Pi_{h'}\Pi_h$, which agree. Their common
$+1$ space, together with the $P_j$ checks, is exactly the encoded
outer code. In particular, the construction has retained all $K$
logical eight-level systems.

For the distance lower bound, an operator on fewer than $d_0$ binary
qubits touches fewer than $d_0$ cells. Compressing it through the cell
isometries gives an operator on at most those outer clocks, which
the outer code detects. This argument also covers errors that take
a cell outside its valid subspace: compression computes the matrix
elements between code states without assuming that the error preserves
the cell. Consequently $V^\dagger FV=c_FI_{LG}$ for every such
physical operator $F$, proving distance at least $d_0$ for both the
full subspace and the protected subsystem.

To prove equality, use the logical phase operator obtained by applying
$\tau$ to $(4ce_j^T,0)$. It has exactly $d_0$ nonzero physical
coefficients, all equal to four. Its nontrivial logical action has
order two, so in the free logical phase basis it is a nonidentity
product of $Z_{8,i}^4$. Under $x=a+2b$, these factors act as
$(-1)^a$: they are protected logical $Z$ operators and act as
identity on the gauge. In the binary cell, the first bit of the
$4x$ block is precisely $x\bmod2$, since $[4x]_8$ is zero or four.
A single physical $Z$ on that bit therefore implements $Z_8^4$.
The outer witness consequently has a binary representative of weight
$d_0$ acting nontrivially on the protected register. This proves both
exact distance claims in \cref{thm:protected}.

The construction retains sparse commuting checks because it keeps the
full logical eight-level systems. Restricting each one to two logical
states, as in \cref{sec:clocks}, would require additional logical
constraints whose support need not be bounded. Here those extra states
instead form the gauge register, with the specified Clifford evolution.
The binary projectors are generally non-Pauli, so this is a construction
with sparse commuting checks rather than a Pauli-LDPC realization.

\Needspace{7\baselineskip}
\section{Future Work}\label{sec:future}

The main open intersection is an asymptotically good Pauli-LDPC family with
fully addressable transversal non-Clifford one-qubit gates. Even constant
rate with sublinear growing distance remains unachieved here for the
fully addressable $T$ interface and ordinary Pauli checks.

Several intermediate questions are concrete. Can a Pauli-LDPC family carry
all logical $T$ phases in one common bounded-participation set of physical
$T$, controlled-$S$, and CCZ terms? Expressing their phases through parity
ancillas using \eqref{eq:controlled} suggests a route to transversal $T$ powers;
the challenge is to
keep one common sparse presentation for all logical targets. Can the sparse protected
construction reach linear distance, or select a useful fixed logical sector
while retaining sparse checks? Finally, its quantum decoding and check
measurement costs remain to be analyzed.

\subsection*{Generative AI Statement}
The basic idea of Sections 3 and 5 were conceived by the author. These were developed and validated with the assistance of ChatGPT (a mix of 5.6-Sol and 6-Astra, across several sessions). The process of discussion and refinement led to ChatGPT proposing the construction of Section 4 as an additional attractive result. Several attempts were made with AI to achieve more of the desiderata given in Section 6 simultaneously, including several other routes the author was optimistic on, but without luck. These were completed by September 5th. The initial draft of writing was done by the author but significantly rewritten and replaced by AI generated text through the editing process. The author takes responsibility for the correctness of claims and proofs and has vetted all text herein.

\appendix
\section{Quantitative comparison for addressable \texorpdfstring{$T$}{T}}
\label{app:parameters}

Compare the full literal interface \eqref{eq:literal}--\eqref{eq:strong},
writing $R=\liminf k/n$ and $\delta=\liminf d/n$.
For San-Jos\'e's hierarchy-lowering route~\cite[Propositions~3.2 and~4.5,
Theorem~6.4, Remark~6.8]{SanJose}, take $h=4$, $m=18$, and $\tau=128$.
The generic bounds give $s\le106761$ and $B=257/65408$.
Repetition applied directly to the weighted code uses
$r_j\in\{7,9,11,13\}$ copies, matching each odd coefficient modulo eight
and exposing every $T$ power and the uniform target.
The bounds $R\ge\kappa/(13s)$ and $\delta\ge(B-7\kappa)/(13s)$,
at $\kappa=4B/29$, together with \eqref{eq:constants}, certify
$(R,\delta)\ge(4\eta,\eta)$ coordinatewise, with
\[
 \eta_{\rm here}=\frac1{195118}\simeq5.13\times10^{-6},\qquad
 \eta_{\rm SJ}=\frac{257/65408}{29\cdot13\cdot106761}
       \simeq9.76\times10^{-11}.
\]
Thus both displayed guarantees (relative distance, rate) improve by about
$5.2\times10^4$ at the same rate-to-distance ratio. These are lower bounds
and not meaningfully optimized, but are meant to give only a crude indication
of the different effective parameters of each construction. Without a doubt,
both can be greatly improved in practice.

\end{document}